\documentclass[10pt,a4paper,onecolumn]{IEEEtran}

\usepackage{color}

\usepackage{amsfonts,bm}

\usepackage{amssymb}
\usepackage{amsthm}
\usepackage{mathtools}

\newtheorem{theorem}{Theorem}
\newtheorem{proposition}[theorem]{Proposition}
\newtheorem{lemma}[theorem]{Lemma}
\newtheorem{corollary}[theorem]{Corollary}
\newtheorem{conjecture}[theorem]{Conjecture}
\theoremstyle{remark}
\newtheorem{remark}[theorem]{Remark}
\theoremstyle{definition}
\newtheorem{definition}[theorem]{Definition}
\newtheorem{example}[theorem]{Example}
\newtheorem{construction}[theorem]{Construction}

\newcommand{\Irred}{{\mathcal{I}}}
\newcommand{\SidonSet}{{\mathcal{S}}}
\newcommand{\qbin}[3]{{{\left[\genfrac{}{}{0pt}{}{#1}{#2}\right]}_{#3}}}
\newcommand{\grsmn}[3]{{\mathcal{G}_{#1}{\left(#2,#3\right)}}}

\newcommand{\Span}[1]{{\left\langle {#1} \right\rangle}}
\renewcommand{\mathbf}[1]{{\bm{#1}}}     
\newcommand{\blda}{{\mathbf{a}}}
\newcommand{\bldb}{{\mathbf{b}}}
\newcommand{\bldv}{{\mathbf{v}}}
\newcommand{\bldx}{{\mathbf{x}}}
\newcommand{\bldxi}{{\mathbf{\xi}}}

\newcommand{\trace}{{\mathrm{tr}}}
\newcommand{\code}{{\mathcal{C}}}
\newcommand{\field}{{\mathbb{F}}}
\newcommand{\integers}{{\mathbb{Z}}}
\newcommand{\orb}{{\mathrm{orb}}}
\newcommand{\transpose}{{\mathsf{T}}}
\newcommand{\event}{{\mathcal{E}}}
\newcommand{\Prob}{{\mathsf{Prob}}}
\newcommand{\Quad}{{\mathcal{Q}}}

\newcommand{\ceil}[1]{\left\lceil{#1}\right\rceil}
\newcommand{\floor}[1]{\left\lfloor{#1}\right\rfloor}

\renewcommand*{\thefootnote}{\fnsymbol{footnote}}

\begin{document}
	
	\title{Construction of Sidon Spaces with\\Applications to Coding}
	
	\author{\textbf{Ron M.~Roth$^\star$}, \textbf{Netanel Raviv$^{\star,\dagger}$}, and \textbf{Itzhak Tamo$^\dagger$}\\
		\IEEEauthorblockA{$^\star$Computer Science Department, Technion -- Israel Institute of Technology, Haifa 3200003, Israel\\
			$^\dagger$Department of Electrical Engineering--Systems, Tel-Aviv University, Tel-Aviv, Israel\\
			\textit{ronny@cs.technion.ac.il, netanel.raviv@gmail.com, zactamo@gmail.com}}
	}
	
	\maketitle
	\thispagestyle{empty}
	\IEEEpeerreviewmaketitle
	
	\begin{abstract}
		A subspace
of a finite extension
field is called a Sidon space if the product of any two of its elements is unique up to
a scalar multiplier from the base field.
Sidon spaces were recently introduced by Bachoc \emph{et al.} as a means to characterize multiplicative properties of subspaces, and yet no explicit constructions were given. In this paper, several constructions of Sidon spaces are provided. In particular, in some of the constructions the relation between~$k$, the dimension of the Sidon space, and~$n$, the dimension of the ambient
extension field,
is optimal.
		
		These constructions are shown to provide cyclic subspace codes,
which are useful tools in network coding schemes.
		To the best of the authors' knowledge, this constitutes the first set of constructions of non-trivial cyclic subspace codes in which the relation between~$k$ and~$n$ is polynomial, and in particular, linear. As a result, a conjecture by Trautmann \emph{et al.} regarding the existence of non-trivial cyclic subspace codes
		is resolved for most parameters, and multi-orbit cyclic subspace codes are attained, whose cardinality is within a constant factor
(close to $1/2$) from the sphere-packing
bound for subspace codes.
	\end{abstract}
	
	\footnotetext{
The work of R.~M.~Roth was supported in part by Grant~1396/16 from
the Israel Science Foundation (ISF), and
the work of I.~Tamo and N.~Raviv was supported by ISF Grant~1030/15 and 
NSF-BSF Grant~2015814. Parts of this paper will be presented at the International Symposium on Information Theory (ISIT), Aachen, Germany, June 2017.}
	
	\begin{IEEEkeywords}
	Sidon spaces, Network coding, Cyclic subspace Codes, Sidon sets.
	\end{IEEEkeywords}
	
	\renewcommand{\thefootnote}{\arabic{footnote}}
	
	\section{Introduction}\label{section:introduction}

	Let $\grsmn{q}{n}{k}$ be the set of
all~$k$-dimensional subspaces of~$\field_{q^n}$,
the degree-$n$ extension field of the finite field~$\field_q = \mathrm{GF}(q)$. \emph{Sidon spaces} were recently defined in~\cite{Vospers} as a tool for studying certain multiplicative properties of subspaces. In particular, this term was used to characterize subspaces~$S$ and~$T$
of~$\field_{q^n}$ such that the subspace
$S\cdot T\triangleq \Span{\{s\cdot t \,:\, s\in S, t\in T\}}$
is of small dimension, where~$\Span{\cdot}$
denotes linear span over~$\field_q$.
Simply put, a Sidon space is a subspace $V \in \grsmn{q}{n}{k}$
such that the product of any two nonzero elements of~$V$
has a unique factorization over~$V$,
up to a constant multiplier from~$\field_q$.
As noted in~\cite{Vospers}, the term ``Sidon space'' draws its inspiration from a \emph{Sidon set}. A set of integers is called a Sidon set if the sums of any two (possibly identical) elements in it are distinct;
thus, Sidon spaces may be seen as a multiplicative and linear variant of Sidon sets.
A formal definition follows; hereafter, for $a\in\field_{q^n}$,
the notation $a\field_q$ stands for
$\{\lambda a \,:\, \lambda\in\field_q \}$.
	
	\begin{definition}[{\cite[Sec.~1]{Vospers}}]\label{definition:SidonSpace}
		A subspace~$V\in\grsmn{q}{n}{k}$ is called a Sidon space if for all nonzero~$a,b,c,d\in V$, if~$ab=cd$ then~$\{a\field_q,b\field_q \}=\{c\field_q,d\field_q \}$.
	\end{definition}

In this paper, we present several constructions of Sidon spaces
using a variety of tools.
The constructions exhibit either a linear or a quadratic relation
between~$n$ and~$k$. When~$n$ is quadratic in~$k$, some of the resulting Sidon spaces satisfy an additional property of linear independence between products of basis
elements.

One of our motivations for studying Sidon spaces is the construction of cyclic subspace codes, which are defined as follows.
A (constant dimension)
subspace code is a subset of~$\grsmn{q}{n}{k}$ under the subspace metric $\mathsf{d_S}(U,V)\triangleq \dim U+\dim V-2\dim(U\cap V)$. The interest in subspace codes has increased recently due to their application to error correction in random network coding~\cite{CodingFor}. In order to find good subspace codes and to study their structure, \emph{cyclic subspace codes} were introduced~\cite{EtzionVardy}.
	For a subspace $U\in\grsmn{q}{n}{k}$ and a nonzero element~$\alpha\in\field_{q^n}^*\triangleq\field_{q^n}\setminus\{0\}$, the \emph{cyclic shift} of~$U$ by~$\alpha$ is
$\alpha U\triangleq{\{\alpha \cdot u \,:\, u\in U\}}$,
which is clearly a subspace of the same dimension as~$U$.
The \emph{orbit} of~$U$ is
$\orb(U)\triangleq\{\beta U \,:\, \beta\in\field_{q^n}^*\}$,
and its cardinality is $(q^n{-}1)/(q^t{-}1)$
for some integer~$t$ which divides~$n$. A subspace code is called \emph{cyclic} if it is closed under cyclic shifts. For example, if~$k$ divides~$n$ then since~$\field_{q^k}^*$ is a multiplicative subgroup of~$\field_{q^n}^*$ whose cosets
are~$\{\alpha \field_{q^k}^* \,:\, \alpha\in\field_{q^n}^*\}$,
it follows that~$\orb(\field_{q^k})$ is a cyclic subspace code of size
$(q^n{-}1)/(q^k{-}1)$
and minimum distance\footnote{Unless otherwise stated, the term ``minimum distance'' refers to the
minimum subspace distance according to
the metric~$\mathsf{d_S}(\cdot,\cdot)$.}
$2k$ (when $k \,|\, n$, this size is, in fact, the largest possible
for any subspace code $\code \subseteq \grsmn{q}{n}{k}$ of
that minimum distance~\cite[Sec.~3]{EtzionVardy}).
	
	Besides the aforementioned example, no general construction of
cyclic subspace codes
was known until~\cite{SubspacePolynomials} provided a few constructions from orbits of root spaces of properly chosen linearized polynomials\footnote{\label{footnote:linearizedPolynomials}A polynomial of the form $L(x)=\sum_i a_ix^{q^i}$, where~$a_i\in \field_{q^n}$ for all~$i$, is called \emph{a linearized polynomial} with respect to the base field~$\field_q$.}. These constructions have minimum distance of~$2k{-}2$ (i.e., intersection dimension at most~$1$), and orbits of full size
$(q^n{-}1)/(q{-}1)$.
However, the relation between~$n$ and~$k$ in the constructions of~\cite{SubspacePolynomials} is in general not known. The results of~\cite{SubspacePolynomials} were recently improved by~\cite{Otal}, which developed a technique to increase the number of distinct orbits without compromising the minimum distance, and hence to increase the size of the cyclic subspace code. Yet,~\cite{Otal} did not address the problem of
minimizing~$n$ for a given~$k$.
The following conjecture was
posed\footnote{Notice that~\cite{CyclicOrbitCodes}
uses the more general term \emph{cyclic orbit codes},
of which cyclic subspace codes are a special case.}
in~\cite{CyclicOrbitCodes}
regarding the relation between~$n$ and~$k$,
	
	\begin{conjecture}[{\cite[Sec.~IV.D]{CyclicOrbitCodes}}]\label{conjecture:cyclicOrbitCodes}
		For any prime power~$q$ and positive integers~$k$ and~$n\ge 2k$, there exists a cyclic subspace code~$\code \subseteq \grsmn{q}{n}{k}$ of minimum distance~$2k{-}2$ and cardinality~$(q^n{-}1)/(q{-}1)$.
	\end{conjecture}

	As implied by Lemma~\ref{lemma:optimalK} below, the requirement~$2k\le n$ in Conjecture~\ref{conjecture:cyclicOrbitCodes} is necessary. In this paper,
we show that constructing a single orbit cyclic subspace code of
maximum size~$(q^n{-}1)/(q{-}1)$
and minimum distance~$2k{-}2$ in~$\grsmn{q}{n}{k}$ is equivalent to constructing a Sidon space in~$\grsmn{q}{n}{k}$, a fact which is also shown in~\cite{Vospers}. For~$q\ge 3$, our constructions of Sidon spaces resolve Conjecture~\ref{conjecture:cyclicOrbitCodes} for any~$k$ and
even~$n\ge 2k$;
for~$q = 2$, they resolve the conjecture for any~$k$ up to
the largest divisor of~$n$ that is smaller than~$n/2$.
In some cases, a simple generalization allows to include in the code multiple orbits of distinct Sidon spaces without compromising the minimum distance; in these cases, the cardinality of the resulting code is within a constant factor
(close to~$1/2$) from the sphere-packing
bound for subspace codes~\cite{EtzionVardy}.
	
	The Sidon space constructions in this paper also resolve an open question regarding the
\emph{square span} (in short, the \emph{square})
of a Sidon space. For a subspace $V\in\grsmn{q}{n}{k}$, the square of~$V$ is defined as the span of all products of pairs of elements from~$V$, i.e., $V^2\triangleq \Span{\{uv \,:\, v,u\in V\}}$.
In~\cite[Thm.~18]{Vospers} it is proved that $\dim(V^2)\ge 2\dim V$ for any Sidon space of dimension~3 or more in~$\field_{q^n}$.
Since~$V^2$ is spanned by $\binom{k+1}{2}$ elements of~$\field_{q^n}$,
we get the following lower and upper bounds on~$\dim(V^2)$.
	
	\begin{proposition}\label{proposition:dimV^2}
		If~$V\in\grsmn{q}{n}{k}$ is a Sidon space then,
whenever $k \ge 3$,
\[
	2k \le \dim(V^2)\le \binom{k+1}{2} .
\]
	\end{proposition}

In light of Proposition~\ref{proposition:dimV^2}, we study Sidon spaces in both possible endpoints and, hence, we introduce the following terms.

\begin{definition}\label{definition:minmax-spanSidonSpace}
	A subspace~$V\in\grsmn{q}{n}{k}$
	is called a \emph{minimum-span} (in short, \emph{min-span}) Sidon space
	if it is a Sidon space and, in addition,~$\dim(V^2)=2k$;
	i.e.,~$V$ attains the lower bound
	in Proposition~\ref{proposition:dimV^2}. Similarly, a subspace~$V\in\grsmn{q}{n}{k}$ is called a \emph{maximum-span} (in short, \emph{max-span}) Sidon space if it is a Sidon space and, in addition,~$\dim(V^2)=\binom{k+1}{2}$; i.e.,~$V$ attains the upper bound in Proposition~\ref{proposition:dimV^2}.
\end{definition}

	Finding the smallest
possible value of~$\dim(V^2)$ for a given~$k$ (in particular, deciding 
if min-span Sidon spaces exist),
is listed as an open problem in~\cite[Sec.~8]{Vospers}.
In Section~\ref{section:Constructions}, we show
that the answer is affirmative
whenever~$k$ is a proper divisor of~$n$
(smaller than $n/2$ when $q = 2$).
Then, in Section~\ref{section:max-span}, we present
constructions of max-span Sidon spaces.
In Section~\ref{section:applications}, we exhibit
the connection between Sidon spaces in~$\grsmn{q}{n}{k}$
and cyclic subspace codes of minimum distance~$2k{-}2$,
thereby proving Conjecture~\ref{conjecture:cyclicOrbitCodes}
for any~$k$ and even~$n\ge 2k$ when~$q \ge 3$
(for~$q = 2$, we prove the conjecture for any~$k$ up to
the largest divisor of~$n$ that is smaller than~$n/2$).
In addition, we show that Sidon spaces imply Sidon sets, and hence, multiple novel constructions of Sidon sets are obtained.
Finally, in Section~\ref{section:rSidon}, we introduce and construct~$r$-Sidon spaces,
which are natural generalizations
of Sidon spaces in which the product of any~$r$ elements is unique up to a
scalar multiplier from the base field.
	
	\section{Preliminaries}

	We start with a lemma that provides a necessary condition on~$n$ and~$k$ so that~$\grsmn{q}{n}{k}$ contains a Sidon space. 
	
	\begin{lemma}\label{lemma:optimalK}
		If~$V\in\grsmn{q}{n}{k}$ is a Sidon space then
$2k\le n$ whenever~$k\ge 3$.
	\end{lemma}
	\begin{proof}
On the one hand, by Proposition~\ref{proposition:dimV^2}
we have~$2k=2\dim(V)\le \dim(V^2)$. On the other hand,~${\dim(V^2)\le n}$. 
	\end{proof}

	The case of Sidon spaces of dimension~$k\in\{1,2\}$ is
straightforward and is treated
in Appendix~\ref{section:SidonOfDimOneTwo}. In the remainder of this paper, it is assumed that~$k\ge 3$.
	
	\begin{remark}\label{remark:smallerK}
		It is readily verified that any subspace of a Sidon space is a Sidon space. Hence, a construction of a Sidon space in~$\grsmn{q}{n}{k}$ implies the existence of Sidon spaces
in~$\grsmn{q}{n}{t}$, for any~$1\le t\le k$.\qed
	\end{remark}
	
Sidon spaces are closely related to their namesakes Sidon sets:
the latter
are shown in the sequel to both imply and to be implied by Sidon spaces 
(Sections~\ref{section:max-span} and~\ref{section:SidonSets}).
Hereafter, $[m]$ stands for the integer set $\{ 1, 2, \ldots, m \}$.
	\begin{definition}\label{definition:SidonSet}
		A subset $\{n_1,n_2,\ldots,n_k\}$ of an Abelian group~$G$ is called a Sidon set if all pairwise sums are unique, i.e., if
the set $\{n_i+n_j \,:\, i,j\in[k], i\ge j\}$
contains $\binom{k+1}{2}$ distinct elements.
	\end{definition}

	In the sequel, the group~$G$ in Definition~\ref{definition:SidonSet} is either the set~$\integers$ of integers or the set~$\integers_m$ of integers modulo some natural number~$m>1$. Clearly, a Sidon set in~$\integers_m$ is a Sidon set in~$\integers$, but not necessarily vice versa. The main challenge in constructing Sidon sets is obtaining high \emph{density}, i.e., having~$m$ as small as possible
given the size~$k$ for Sidon sets in~$\integers_m$, or alternatively,
for integer Sidon sets in~$[m]$.
	
	Sidon sets have attracted considerable attention over the years, and many constructions are known~\cite{SidonSurvey}. In particular, there exist several constructions in which $m=k^2(1+o_k(1))$ (where~$o_k(1)$ stands for a term that goes to~$0$ as~$k$ goes to infinity);
moreover, as noted in~\cite[Sec.~4.1]{SidonSurvey},
such a relation between~$m$ and~$k$ is optimal
(up to a factor of $1 + o_k(1)$). A known example of Sidon sets is quoted next.
	
	\begin{example}[{\cite{Bose}}]\label{example:Bose}
For a prime power~$q$, a primitive element~$\gamma$ in~$\field_{q^2}$,
and an arbitrary element $\delta\in\field_{q^2}\setminus \field_q$,
the set~$\{\log_\gamma(\alpha+\delta) \,:\, \alpha\in\field_q \}$
is a Sidon set of size~$q$ in $\integers_{q^2-1}$.\qed
	\end{example}

One of our motivations for studying Sidon spaces is the construction of cyclic subspace
codes, which were defined in~\cite{EtzionVardy} as follows.

\begin{definition}\label{definition:CyclicSubspaceCodes}
	A subspace code
$\code\subseteq\grsmn{q}{n}{k}$
is cyclic\footnote{Definition~\ref{definition:CyclicSubspaceCodes}
is not to
be confused with the more general term ``(cyclic) orbit codes''~\cite{CyclicOrbitCodes}, in which an extension-field structure of the ambient space is not assumed. Instead, the ambient space of~$\grsmn{q}{n}{k}$ in~\cite{CyclicOrbitCodes} is~$\field_q^n$, the vector space of dimension~$n$ over~$\field_q$; and a (cyclic) orbit code is any subspace code which is closed under the action of a (cyclic) subgroup of $\textrm{GL}_n(q)$, the group of invertible~$n\times n$ matrices over~$\field_q$. To be precise, a cyclic code may be seen as a cyclic orbit code which is closed under the action of
a cyclic subgroup of $\textrm{GL}_n(q)$ that is
isomorphic to~$\field_{q^n}^*$.}
if~$V\in\code$ implies that $\alpha V\in\code$,
for every~$\alpha\in\field_{q^n}^*$.
\end{definition}

Nearly optimal \emph{non}-cyclic
subspace codes are known to exist for a wide set of parameters. Specifically, the cardinality of the so-called Koetter--Kschischang codes~\cite[Sec.~V]{CodingFor} is within a factor of~$1+o_q(1)$ from the known upper bounds on the size of subspace
codes; one of these bounds is quoted next, where we use
the notation~$\qbin{t}{s}{q}$
for the \emph{$q$-binomial coefficient}
(also known as the Gaussian coefficient) $|\grsmn{q}{t}{s}| = \prod_{i=0}^{s-1}((q^{t-i}{-}1)/(q^{i+1}{-}1))$.

\begin{theorem}[Sphere-packing bound for subspace codes %
{\cite[Thm.~2]{EtzionVardy}}]\label{theorem:EtzionVardyBound}
	A subspace code~$\code\subseteq \grsmn{q}{n}{k}$
with minimum distance~$d$
satisfies
\[
	|\code|\le \frac{\qbin{n}{k-d/2+1}{q}}{\qbin{k}{k-d/2+1}{q}}.
\]
\end{theorem}
For the parameters~$n=2k$ and~$d=2k{-}2$, Koetter--Kschischang codes are of size~$q^n$, whereas the upper bound of Theorem~\ref{theorem:EtzionVardyBound} is~$q^n(1+o_n(1))$. For these parameters, the largest codes
presented
in this paper are within a factor of~$1/2+o_q(1)$ from this bound. 

\section{Construction of Minimum-Span Sidon Spaces}\label{section:Constructions}

This section begins with two similar constructions of Sidon spaces in which~$n$
grows linearly with~$k$.
The first construction applies to any prime power~$q$ and~$n=rk$,
for an integer~$r\ge 3$. The second construction applies to~$q\ge 3$ and~$n=2k$, and is therefore (fully) optimal by Lemma~\ref{lemma:optimalK}. 

\begin{construction} \label{construction:MainBinary}
	For a composite integer~$n$, let~$k$ be the largest divisor of~$n$ which is smaller than~$n/2$, let~$\gamma\in\field_{q^n}^*$ be a root of an irreducible polynomial
of degree~$n/k$ over~$\field_{q^k}$,
and let $V\triangleq\{u+u^q\gamma \,:\, u\in\field_{q^k} \}$.
\end{construction}

While
Construction~\ref{construction:MainBinary}
requires the dimension~$k$ of~$V$ to be a divisor of~$n$
(and~$2k < n$), once we establish that~$V$ is a Sidon space,
we will get, by Remark~\ref{remark:smallerK}, Sidon spaces for
all dimensions that are smaller than~$k$.
The fact that~$V$ in Construction~\ref{construction:MainBinary} is
indeed
a Sidon space is shown in the next theorem.

\begin{theorem}\label{theorem:MainBinaryProof}
	The subspace~$V\in \grsmn{q}{n}{k}$ from Construction~\ref{construction:MainBinary} is a Sidon space.
\end{theorem}

\begin{proof}
	Given a product $(u + u^q \gamma)(v + v^q \gamma)$, for some nonzero~$u$ and~$v$ in~$\field_{q^k}$, we notice that
	\begin{equation}\label{equation:MainBinary}
	(u+u^q\gamma)(v+v^q\gamma)=uv+(uv^q+u^qv)\gamma+(uv)^q\gamma^2.
	\end{equation}
	Since~$n>2k$, it follows that~$\{1,\gamma,\gamma^2\}$ is a linearly independent set over~$\field_{q^k}$, thus the
coefficients of the polynomial
$P(x) = (u+u^q x)(v+v^q x) = uv+(uv^q+u^qv) x +(uv)^q x^2$
can be extracted from the right-hand side
of~\eqref{equation:MainBinary}.
The roots of~$P(x)$ are~$-1/u^{q-1}$ and~$-1/v^{q-1}$, from
which the set~$\{ u\field_q, v\field_q \}$ is determined uniquely\footnote{A simple way to compute~$\{u\field_{q} \}$ (say) is representing the equation~$x^q-u^{q-1}x=0$ (over~$\field_{q^k}$) as a set of~$k$ linear homogeneous equations over~$\field_q$, according to some basis of~$\field_{q^k}$ over~$\field_q$. }.
\end{proof}

By choosing the parameter~$\gamma$ in Construction~\ref{construction:MainBinary} judiciously, we can also cover the case~$n=2k$, as long as~$q\ge 3$. To this end, let~$W_{q-1}$ be the set of $(q{-}1)$st powers of elements in~$\mathbb{F}_{q^k}$, i.e,
$W_{q-1}\triangleq \{y^{q-1} \,:\, y\in \field_{q^k}\}$,
and let~$\overline{W}_{q-1}\triangleq \field_{q^k}\setminus W_{q-1}$.

Our next construction will require a monic irreducible
quadratic polynomial over~$\field_{q^k}$ whose free coefficient is in~$\overline{W}_{q-1}$. 
The existence of such polynomials follows from well known
properties of quadratic polynomials over finite fields
(e.g., \cite[Problem~3.42]{Ronny}
and \cite[Problem~3.52]{FiniteFields}), and there is a simple
recipe for constructing such polynomials explicitly.
Recall
that for odd~$q$, the elements of
$\textup{QR}(q^k) \triangleq \{y^2 \,:\, y\in\field_{q^k}^* \}$
are called \emph{quadratic residues}
of~$\field_{q^k}$, and the elements
of~$\textup{QNR}(q^k) = \field_{q^k}^* \setminus \textup{QR}(q^k)$
are called \emph{quadratic non-residues}
of~$\field_{q^k}$. In addition, for an even~$q$ and an element~$y\in \field_{q^k}$ let $\trace(y)\triangleq y+y^2+y^4+\ldots+y^{q^k/2}$ be the absolute trace polynomial over~$\field_{q^k}$. It is known
that~$y\mapsto \trace(y)$
is a linear mapping from~$\field_{q^k}$ to~$\field_2$, and, in particular, half of the elements of~$\field_{q^k}$ have trace~$1$.

\begin{lemma}[{\cite[Problem~3.42]{Ronny}}]
\label{lemma:IrreducibleWithGoodC}
	For any~$c\in \field_{q^k}^*$ and any $b\in\field_{q^k}$, the polynomial
$x^2+bx+c$
is irreducible if and only if
	\begin{align*}
	\begin{cases}
	b^2-4c \in \textup{QNR}(q^k),& \textrm{if }q\textrm{ is odd,}\\
	b\ne 0\textrm{ and }\trace\left( c/b^2 \right)=1, &\textrm{if }q\textrm{ is even.}
	\end{cases}
	\end{align*}
\end{lemma}

\begin{corollary}\label{corollary:irrExists}
	For any prime power~$q \ge 2$, positive integer~$k$,
and any~$c\in\field_{q^k}^*$ (in particular, any~$c\in\overline{W}_{q-1}$),
there exists~$b\in\field_{q^k}$ such that $x^2+bx+c$ is irreducible.
\end{corollary}

\begin{proof}
By~\cite[Problem~3.23]{Ronny},
for odd~$q$ and any given~$c\in\field_{q^k}^*$,
the set~$\{b^2-4c \,:\, b\in\field_{q^k} \}\cap \textup{QNR}(q^k)$
is nonempty.
For an even~$q$ we have
$\{ c/b^2 \,:\, b\in \field_{q^k}^*\}=\field_{q^k}^*$
for any~$c\in\field_{q^k}^*$, and the claim follows by the properties of the mapping~$y\mapsto \trace(y)$.
\end{proof}

\begin{construction}\label{construction:mainNonBinary}
	For a prime power~$q\ge 3$ and a positive integer~$k$, let~$n=2k$, let $\gamma\in\field_{q^n}^*$ be a root of an irreducible polynomial $x^2+bx+c$ over~$\field_{q^k}$ with $c\in \overline{W}_{q-1}$, and
let~$V\triangleq \{u+u^q\gamma \,:\, u\in\field_{q^k}\}$.
\end{construction}

\begin{theorem}\label{theorem:mainNonBinary}
	The subspace~$V\in\grsmn{q}{n}{k}$ from Construction~\ref{construction:mainNonBinary} is a Sidon space.
\end{theorem}

\begin{proof}
	As in the proof of Theorem~\ref{theorem:MainBinaryProof}, it suffices to show that given $(u+u^q\gamma)(v+v^q\gamma)$ for some nonzero~$u$ and~$v$ in~$\field_{q^k}$, the set
$\{u \field_q, v \field_q\}$
can be determined uniquely. Since
	\begin{equation}\label{equation:mainNonBinary}
		(u+u^q\gamma)(v+v^q\gamma)={\underbrace{\left(uv-(uv)^qc\right)}_{\triangleq \, Q_0}}+{\underbrace{\left(uv^q+u^qv-b(uv)^q\right)}_{\triangleq \, Q_1}} \gamma,
	\end{equation}
	and since $\{1,\gamma\}$ is a linearly independent set over~$\field_{q^k}$, it follows that given the product $(u+u^q\gamma)(v+v^q\gamma)$, one may easily
extract~$Q_0$ and~$Q_1$ (both in~$\field_{q^k}$) from
the right-hand side of~\eqref{equation:mainNonBinary}. Now, notice
that~$Q_0$
is the value of the linearized polynomial $T(x)=x-cx^q$ at the point~$uv$. We next show that the mapping~$x\mapsto T(x)$ (which is linear over~$\field_q$) is invertible on~$\field_{q^k}$, i.e., $T(x)=0$ only if $x=0$. 
	
	Indeed, if there were a nonzero~$\beta\in\field_{q^k}$ such
that~$T(\beta)=0$,
then~$\beta-c\beta^q=0\Rightarrow c=\beta^{-(q-1)}\Rightarrow c\in W_{q-1}$, a contradiction.
	Thus,
given~$Q_0$,
it is possible to uniquely determine~$uv$ by solving a set of~$k$ linear equations over~$\field_q$ (the
equations that represent~$T(x)=Q_0$
according to some basis of~$\field_{q^k}$ over~$\field_q$). From~$uv$
and~$Q_1$ it is possible to compute
the coefficients of the polynomial $P(x)$ defined
in the proof of Theorem~\ref{theorem:MainBinaryProof},
and, as in that proof, the roots of $P(x)$
determine~$\{u\field_q,v\field_q \}$ uniquely.
\end{proof}

Theorems~\ref{theorem:MainBinaryProof} and~\ref{theorem:mainNonBinary} resolve an open question from~\cite{Vospers} regarding the existence of min-span Sidon spaces.
Since~$\dim V=k=n/2$
in Construction~\ref{construction:mainNonBinary}, and since $\dim(V^2)\le n$, it follows that Construction~\ref{construction:mainNonBinary}
attains the lower bound in Proposition~\ref{proposition:dimV^2},
namely, it is a min-span Sidon space.
Moreover, Construction~\ref{construction:MainBinary}
yields a min-span Sidon space too, now with~$\dim(V^2)$ strictly smaller than~$n$; we show this in the following lemma.

\begin{lemma}
	The subspace~$V$ in Construction~\ref{construction:MainBinary} is a min-span Sidon space.
\end{lemma}

\begin{proof}
	First, observe that~\eqref{equation:MainBinary} implies that
$V^2=\Span{\{uv+(uv^q+u^qv)\gamma+(uv)^q\gamma^2 \,:\, u,v\in V\}}$.
Secondly, since the mapping $A \mapsto A^q$ on~$\field_{q^k}$
is linear over~$\field_q$,
the set $U\triangleq \{A+B\gamma+A^q\gamma^2 \,:\, A,B\in\field_{q^k}\}$
is a linear subspace over~$\field_q$
and its dimension equals~$2k=2\dim V$.
The result follows from the containment~$V^2\subseteq U$.
\end{proof}

\begin{remark}
	Constructions~\ref{construction:MainBinary}	and~\ref{construction:mainNonBinary} can be generalized
to~$V\triangleq\{ u + u^{q^s} \,:\,  u \in \field_{q^k} \}$,
where $\gcd(s,k) = 1$.
To see this, note that in the proofs of
Theorems~\ref{theorem:MainBinaryProof} and~\ref{theorem:mainNonBinary},
the set~$\{ u \field_q \}$ can be uniquely
determined from~$u^{q^s-1}$, for every $u \in \field_{q^k}$;
this is indeed so, since $u^{q^s-1} = (u^{(q^s-1)/(q-1)})^{q-1}$ and the mapping $u \mapsto u^{(q^s-1)/(q-1)}$ is
bijective
on $\field_{q^k}$ and on $\field_q$. By the same reasoning, the mapping $ x \mapsto x - c x^{q^s}$ is
bijective on~$\field_{q^k}$.\qed
\end{remark}

While
Constructions~\ref{construction:MainBinary} and~\ref{construction:mainNonBinary} provide a Sidon space whose dimension scales linearly with~$n$, they do not apply
to
all possible values of~$n$. The following theorem, whose proof is deferred to Appendix~\ref{section:SidonForAnyN}, shows that such Sidon spaces exist for any~$n$.

\begin{theorem}\label{theorem:Existence}
	For any prime power~$q$ and integer~$n \ge 6$, there exists a Sidon space in $\grsmn{q}{n}{\floor{(n{-}2)/4}}$.
\end{theorem}

\section{Construction of Maximum-Span Sidon Spaces}\label{section:max-span}

The constructions we presented in Section~\ref{section:Constructions} attain the lower bound in
Proposition~\ref{proposition:dimV^2}.
In this section, we consider the other extreme
case---namely, constructions of max-span Sidon spaces,
which attain the upper bound in that proposition.
The next lemma states that if a subspace attains the upper bound in
Proposition~\ref{proposition:dimV^2},
then the subspace is in effect, a Sidon space.

\begin{lemma}\label{lemma:max-spanSidonSpace}
	For a subspace~$V$ in~$\grsmn{q}{n}{k}$,
if $\dim(V^2)=\binom{k+1}{2}$,
then~$V$ is a (max-span) Sidon space.
\end{lemma}

\begin{proof}
	In light of Definition~\ref{definition:SidonSpace}, it suffices to prove that if
$a,b,c,d\in V$ satisfy $ab=cd \ne 0$,
then~$\{ a\field_q,b\field_q \}= \{ c\field_q,d\field_q \}$.
To this end, let
$\bldv=(v_i)_{i \in [k]}$ be a vector over~$\field_{q^k}$
whose entries form a basis of~$V$ over~$\field_q$, let
$a,b,c,d\in V \setminus \{ 0 \}$,
and denote
	\begin{align*}
	a &= \sum_{i \in [k]} a_i\cdot v_i= p_a(\bldv),
        &b = \sum_{i \in [k]} b_i\cdot v_i= p_b(\bldv),\\
        c &= \sum_{i \in [k]} c_i\cdot v_i= p_c(\bldv),
        &d =\sum_{i \in [k]} d_i\cdot v_i= p_d(\bldv),
	\end{align*}
	where $a_i,b_i,c_i,$ and $d_i$ are elements of $\field_q$ for
all $i\in[k]$, and $p_a,p_b,p_c$, and~$p_d$ are
the following multivariate polynomials over $\field_q$ in
the indeterminates $\bldx = (x_i)_{i \in [k]}$:
	\begin{align*}
	p_a(\bldx)&\triangleq \sum_{i\in[k]}a_ix_i,
	\hspace{2cm}p_b(\bldx)\triangleq \sum_{i\in[k]}b_ix_i,\\
	p_c(\bldx)&\triangleq \sum_{i\in[k]}c_ix_i,
	\hspace{2cm}p_d(\bldx)\triangleq \sum_{i\in[k]}d_ix_i.
	\end{align*}
	Notice that $ab=cd$ implies that
	\begin{equation}\label{eqn:abcdPolynomials}
	p_a(\bldv)\cdot p_b(\bldv)=\sum_{i\in[k]}a_ib_i v_i^{2}+\sum_{i\ne j}a_ib_j v_i v_j= \sum_{i\in[k]}c_id_i v_i^2+\sum_{i\ne j}c_id_j v_i v_j=p_c(\bldv)\cdot p_d(\bldv).
	\end{equation}
Since~$\dim(V^2)=\binom{k+1}{2}$, and since~$V^2$ is spanned by
the set~$\{v_i\cdot v_j \,:\, i,j\in[k], i\geq j  \}$, it follows that 
the latter set
is linearly independent over~$\field_q$. Hence, we may compare coefficients in~\eqref{eqn:abcdPolynomials} and obtain that
	\begin{align}\label{eqn:PolyCoefficientsEqual}
	\begin{array}{ll}
	a_ib_i=c_id_i		     &\textrm{for all }i\in[k], \textrm{ and }\\
	a_ib_j+a_jb_i=c_id_j+c_jd_i&\textrm{for all distinct }i,j\in[k].
	\end{array}
	\end{align}

	According to~\eqref{eqn:PolyCoefficientsEqual}, it is readily verified that $p_a(\bldx)p_b(\bldx)=p_c(\bldx)p_d(\bldx)$. Since the ring of multivariate polynomials over a field is a unique factorization domain, and since~$p_a,p_b,p_c$, and~$p_d$ are irreducible over~$\field_q$, it follows that
the sets $\{ p_a,p_b \}$ and $\{ p_c,p_d \}$ are equal,
up to a multiplication by a nonzero element of~$\field_q$. Hence,~$\{ a\field_q,b\field_b \}=\{ c\field_q,d\field_b \}$.
\end{proof}

Clearly, max-span Sidon spaces exist
in~$\grsmn{q}{n}{k}$ only if~$n\ge \binom{k+1}{2}$.
In the remainder of this section,
three constructions of Sidon spaces are given. The first two are easily 
seen as being of the max-span type, whereas the third has been verified numerically to be so only
for the few parameters that were tested.
Note that Remark~\ref{remark:smallerK} holds also when
the Sidon spaces referred to therein are all of the max-span type.
We start with our first construction.

\begin{construction}\label{construction:SpaceFromSet}
Let~$\SidonSet\triangleq\{n_1,n_2,\ldots,n_k\}\subseteq [m]$ be
a Sidon set in~$\integers$ (see Definition~\ref{definition:SidonSet})
such that~$m=k^2(1 + o_k(1))$~\cite{SidonSurvey}, and for
an integer $n>2m$ and a proper element~$\gamma$ of~$\field_{q^n}$
(that does not belong to any proper subfield of~$\field_{q^n}$),
let $V\triangleq \Span{\{\gamma^{n_i}\}_{i \in [k]}}$.
\end{construction}

\begin{lemma}
\label{lemma:SpaceFromSet}
The subspace $V \in \grsmn{q}{n}{k}$ from
Construction~\ref{construction:SpaceFromSet}
is a max-span Sidon space.
\end{lemma}

\begin{proof}
	According to Lemma~\ref{lemma:max-spanSidonSpace}, it suffices to show that
the set~$\Gamma\triangleq\{\gamma^{n_i+n_j} \,:\, i,j\in[k],i\ge j \}$,
which spans~$V^2$, is linearly independent over~$\field_q$. Since $\SidonSet$ is a Sidon set, it follows that
the exponents of~$\gamma$ in~$\Gamma$ are distinct.
Furthermore, since~$\gamma$ is proper in~$\field_{q^n}$
and $n>2m$, it follows that~$\Gamma$
contains~$\binom{k+1}{2}$ distinct elements
which are linearly independent over~$\field_q$. Hence,~$V$ is
a max-span Sidon space.
\end{proof}

Next, we turn to our second construction of max-span Sidon spaces.

\begin{construction}
\label{construction:irreduciblePolys}
Let~$\Irred\triangleq\{p_{s,t}(x) \,:\, s,t\in[k], s \ge t \}$
be a set of~$\binom{k+1}{2}$ distinct monic irreducible polynomials
over~$\field_q$ and let~$\Delta$ be the largest degree
of any polynomial in~$\Irred$. For any~$i\in[k]$, let
\begin{equation}
\label{equation:Aifi's}
f_i(x) \triangleq
\prod_{\substack{(s,t) \in [k] \times [k] \,:\\
s \ge t, \, s \ne i, \, t \ne i}}p_{s,t}(x) ,
\end{equation}
and, for $n> 2\Delta\cdot \binom{k}{2}$
and~$\gamma$ proper in~$\field_{q^n}$,
let~$V\triangleq\Span{\{f_i(\gamma)\}_{i\in[k]}}$.
\end{construction}

\begin{lemma}
\label{lemma:irreduciblePolys}
The subspace $V \in \grsmn{q}{n}{k}$ from
Construction~\ref{construction:irreduciblePolys}
is a max-span Sidon space.
\end{lemma}

The proof of Lemma~\ref{lemma:irreduciblePolys} is given in
Appendix~\ref{section:omittedProofs}, where it is also shown that
Construction~\ref{construction:irreduciblePolys} provides
a max-span Sidon space in~$\grsmn{q}{n}{k}$
when $n$ is (as small as)~$(2 + o_k(1)) \, k^2\log_q k$;
in particular, when $q \ge \binom{k+1}{2}$, one can take
$n = k(k{-}1){+}1$.

The third construction we present in this section is based
on the following result from~\cite{SubspacePolynomials}.

\begin{lemma}[\cite{SubspacePolynomials}]
\label{lemma:subspacePoly}
	For any prime power~$q$ and positive integer~$k$, the root space of the linearized polynomial~$\Lambda(x)\triangleq x^{q^k}+x^q+x$ over~$\field_q$ is a Sidon space in the splitting field~$\field_{q^n}$ of~$\Lambda(x)$.
\end{lemma}

Clearly, Lemma~\ref{lemma:subspacePoly} implies the existence of Sidon spaces for
any~$q$ and~$k$. However, no meaningful (e.g., polynomial)
upper bound was given in~\cite{SubspacePolynomials} for the extension
degree~$n$ of the splitting field of~$\Lambda(x)$.
The upcoming sequence of lemmas lead to a proof that the set~$V$ defined
next is a Sidon space in $\grsmn{q}{k^2{-}1}{k}$.

\begin{construction}\label{construction:linearized}
	For any prime power~$q$ and for any integer~$k > 1$ which is a power of~$q$, let $V$ be the root space of~$x^{q^k}+x^q+x$ in~$\field_{q^{k^2-1}}$.
\end{construction}

\begin{lemma}[{\cite[p.~116, Thm.~3.62]{FiniteFields}}]
\label{lemma:qAssociates}
Let $A(x) = \sum_i a_i x^{q^i}$
and $B(x) = \sum_i b_i x^{q^i}$
be linearized polynomials over an extension field of $\field_q$,
where $A(x) \ne 0$.
Then $A(x)$ divides $B(x)$, if and only if
the respective (ordinary) polynomial
$a(x) = \sum_i a_i x^i$ divides $b(x) = \sum_i b_i x^i$.
\end{lemma}

The following lemma is proved in~\cite[p.~92]{Golomb} only for~$q=2$,
yet the proof extends almost verbatim to any prime power~$q$.

\begin{lemma}[\cite{Golomb}]
\label{lemma:Golomb}
	For~$k=q^m$, the polynomial $x^k+x+1$ over~$\field_q$ divides $x^{k^2-1}-1$.
\end{lemma}

\begin{lemma}
\label{lemma:split}
	For~$k=q^m$, the polynomial
$x^{q^k}+x^q+x$
splits in $\field_{q^{k^2-1}}$.
\end{lemma}

\begin{proof}
	By Lemmas~\ref{lemma:qAssociates} and~\ref{lemma:Golomb},
that polynomial divides~$x^{q^{k^2-1}}-x$
and hence it splits in~$\field_{q^{k^2-1}}$.
\end{proof}

Lemmas~\ref{lemma:subspacePoly}
and~\ref{lemma:split} imply the following result.

\begin{lemma}\label{lemma:linearized}
The subspace~$V$ from Construction~\ref{construction:linearized}
is a Sidon space in~$\grsmn{q}{k^2{-}1}{k}$.
\end{lemma}

For $(q,k,n)\in\{(2,4,15),(2,8,63),(3,3,8),(4,4,15) \}$ we have verified numerically that the root space of~$x^{q^k}+x^q+x$
in~$\field_{q^{k^2-1}}$
is a max-span Sidon space. It remains open whether this property holds for every~$q$ and every~$k=q^m$.

We end this section with an existence result of max-span Sidon spaces,
whenever~$n \ge \binom{k+1}{2}$.

\begin{theorem}\label{theorem:ExistenceMaxSpan}
        For any prime power~$q$ and positive integers~$k$
and $n \ge \binom{k+1}{2}$, all but a fraction of less than
\[
\frac{1}{q-1} \cdot q^{k(k+1)/2 - n}
\]
of the spaces in $\grsmn{q}{n}{k}$ are
max-span Sidon spaces. 
\end{theorem}

The proof of this theorem is given in
Appendix~\ref{section:ExistenceMaxSpan}.
It follows from the theorem 
that for~$q > 2$ and any~$n \ge \binom{k+1}{2}$
(and also for $q = 2$ and any~$n > \binom{k+1}{2}$),
a max-span Sidon space in $\grsmn{q}{n}{k}$
can be found in polynomial expected running time by 
a probabilistic algorithm
which picks at random a space in $\grsmn{q}{n}{k}$ and checks whether 
the pairwise products of its basis elements are linearly independent
over $\field_q$.

\begin{remark}
\label{remark:algebra}
While this work focuses on Sidon spaces in extension fields,
this notion can be defined more generally in commutative algebras.
Specifically, let $(\field_q^n,\ast)$
be the algebra given by the $n$-dimensional vector space
over $\field_q$ when endowed with a commutative bilinear
operation~$\ast$.
Then a $k$-dimensional subspace $V \subseteq \field_q^n$
is a Sidon space in the algebra if~$a \ast b=c \ast d$
for nonzero~$a,b,c,d\in V$
implies~$\{a\field_q,b\field_q \}=\{c\field_q,d\field_q \}$.
Lemma~\ref{lemma:max-spanSidonSpace}
still holds under this generalization
(in the proof of the lemma,
replace all $\field_q^n$-products by~$\ast$).
In~\cite{Cascudo}, a result akin to
Theorem~\ref{theorem:ExistenceMaxSpan} was obtained,
with~$\ast$ taken as the coordinatewise
product of vectors in~$\field_q^n$.
We mention that when $q \ge \binom{k+1}{2}$,
the~$k$ polynomials in~\eqref{equation:Aifi's},
when constructed with a set~$\Irred$
consisting of~$\binom{k+1}{2}$ degree-$1$ polynomials,
generate a max-span Sidon space
in the $\binom{k+1}{2}$-dimensional algebra formed by
the polynomial ring modulo~$\prod_{p(x) \in \Irred} p(x)$.\qed
\end{remark}

\section{Applications of Sidon spaces}\label{section:applications}

\subsection{Cyclic subspace codes}\label{section:CyclicSubspaceCodes}

In this subsection,
we show that the orbit of a Sidon space is a cyclic subspace code of minimum distance~$2k{-}2$ and
cardinality~$(q^n{-}1)/(q{-}1)$.
This topic is discussed briefly in~\cite[Sec.~5]{Vospers}, and yet full proofs are provided below for completeness.

This connection between Sidon spaces and cyclic subspace codes, in conjunction with some of the constructions from Section~\ref{section:Constructions}, proves Conjecture~\ref{conjecture:cyclicOrbitCodes} for most parameters. Further, we show that in some cases, orbits of distinct Sidon spaces can be joined into one subspace code without compromising the minimum distance. This fact enables the construction of a cyclic subspace code whose cardinality is within a factor
of~$1/2+o_q(1)$ from
the upper bound of Theorem~\ref{theorem:EtzionVardyBound}.

\begin{lemma}\label{lemma:alphaU=U}
	For $V\in\grsmn{q}{n}{k}$, the following two conditions are equivalent.
	\begin{list}{}{\settowidth{\labelwidth}{\textit{(ii)}}}
		\item[(i)] $|\orb(V)|=(q^n{-}1)/(q{-}1)$.
		\item[(ii)] An element~$\alpha\in\field_{q^n}^*$ satisfies~$\alpha V=V$, if and only if~$\alpha\in \field_{q}^*$.
	\end{list}
\end{lemma}

\begin{proof}
	Since~$\alpha V=(\lambda\alpha)\cdot V$ for any~$\alpha\in\field_{q^n}$ and~$\lambda\in\field_{q}^*$, we can view the elements~$\alpha$
in~(ii)
as if they are elements of the quotient
group~$G\triangleq {\field_{q^n}^*}/{\field_q^*}$.
Let~$H$ be a subgroup of~$G$ which contains all elements~$\beta$ such that~$\beta V=V$.
Then~$|\orb(V)|=|{G}/{H}|$ and, so,
(i)~holds if and only if~$H$ is trivial
(i.e., if and only if~$H$ contains only the unit element).
On the other hand,~$H$ being trivial is equivalent to~(ii).
\end{proof}

\begin{lemma}\label{lemma:SidonGivesCyclic}
	For a subspace $V\in\grsmn{q}{n}{k}$, the set~$\orb(V)$ is of
size~$(q^n{-}1)/(q{-}1)$ and minimum distance~$2k{-}2$,
if and only if $V$ is a Sidon space.
\end{lemma}

\begin{proof}
	Suppose that $\orb(V)$ is of
size~$(q^n{-}1)/(q{-}1)$
and minimum distance~$2k{-}2$, and let~$a,b,c,d\in V$ be
such that~$ab=cd \ne 0$. Write $\alpha\triangleq a/d = c/b$,
and notice that $a,c\in V\cap \alpha V$. If $\alpha\notin\field_q$, it follows from Lemma~\ref{lemma:alphaU=U} and from the minimum distance of $\orb(V)$ that~$\dim(V\cap\alpha V)\le 1$, which implies that $\{a\field_q \}= \{c\field_q \}$. Combining with~$ab=cd$ thus yields $\{a\field_q,b\field_q \}=\{c\field_q,d\field_q \}$. On the other hand, if~$\alpha\in\field_q$ then clearly $\{a\field_q \}=\{d\field_q \}$ and, so, $\{a\field_q,b\field_q \}=\{c\field_q,d\field_q \}$.
	
	Conversely, suppose that
either~$|\orb(V)| < (q^n{-}1)/(q{-}1)$
or the minimum distance of~$\orb(V)$ is less than~$2k{-}2$. Then there exists $\alpha\in\field_{q^n}^*\setminus \field_q^*$ such that $\dim({V\cap \alpha V})\ge 2$, which means that there exist linearly independent elements~$a$ and~$c$ in $V\cap \alpha V$ and respective~$b$ and~$d$ in~$V$ such that $a=\alpha d$ and $c=\alpha b$. This, in turn, implies that $ab=cd$ while~$\{a\field_q \}\ne \{ c\field_q \}$. Yet, since~$\alpha\notin\field_q$, we also
have~$\{a\field_q\}\ne\{d\field_q \}$;
hence, $V$ is not a Sidon space.
\end{proof}

Lemma~\ref{lemma:SidonGivesCyclic} and Remark~\ref{remark:smallerK} yield the following corollary. 

\begin{corollary}\label{corollary:ConjectureProved}
	Construction~\ref{construction:MainBinary} proves Conjecture~\ref{conjecture:cyclicOrbitCodes} for
any~$q$,~$n$,
and any~$k$ up to the largest divisor of~$n$ that is smaller
than~$n/2$.
In addition, Construction~\ref{construction:mainNonBinary} proves Conjecture~\ref{conjecture:cyclicOrbitCodes} for any~$q\ge 3$ and any even~$n\ge 2k$.
\end{corollary}

To attain a cyclic subspace code with multiple orbits, the following lemma is given. This lemma may
be seen
as a variant of Lemma~\ref{lemma:SidonGivesCyclic} for distinct orbits.

\begin{lemma}\label{lemma:distinctOrbits}
The following two conditions are equivalent for any
distinct subspaces~$U$ and~$V$ in~$\grsmn{q}{n}{k}$.
\begin{list}{}{\settowidth{\labelwidth}{\textit{(ii)}}}
\item[(i)]
$\dim(U\cap \alpha V)\le 1$, for any~$\alpha\in\field_{q^n}^*$.
\item[(ii)]
For any nonzero $a,c\in U$ and nonzero~$b,d\in V$,
the equality $ab=cd$ implies
that~$\{a\field_q \}=\{c\field_q \}$ and~$\{b\field_q \}=\{d\field_q \}$.
\end{list}
\end{lemma}

\begin{proof}
	Assume that $\dim(U\cap \alpha V)\le 1$ for all~$\alpha\in\field_{q^n}^*$, and let~$a,c\in U$ and~$b,d\in V$ such that $ab=cd\ne0$.
Define~$\beta\triangleq a/d = c/b$,
and notice that $a,c\in U\cap \beta V$. Since $\dim(U\cap \beta V)\le 1$, it follows that~$\{a\field_q \}=\{c\field_q \}$ and, so,~$\{b\field_q \}=\{d\field_q \}$.
	
	Conversely, assume
	that~$\dim(U\cap \alpha V)\ge 2$ for some~$\alpha\in\field_{q^n}^*$, and let~$a$ and~$c$ be two linearly independent elements in~$U\cap \alpha V$. Then, there exist~$b$ and~$d$ in~$V$ such that~$a=\alpha d$ and~$c=\alpha b$. On the one hand~$ab=cd$, yet on the other hand~$\{a\field_q \}\ne\{c\field_q \}$.
\end{proof}

Lemma~\ref{lemma:distinctOrbits} yields a kind of multi-orbit counterpart of the Sidon space property: the product of any two nonzero elements from \emph{distinct orbits} may be uniquely factorized, up to
a scalar multiplier from~$\field_q$.
Using this lemma, Construction~\ref{construction:mainNonBinary} can be extended to multiple orbits as follows.

\begin{construction}\label{construction:multipleOrbits}
	For a prime power~$q\ge 3$ and a positive integer~$k$, let~$w$ be a primitive element
in~$\field_{q^k}$.
For~$c_0\triangleq w$, let $b_0\in\field_{q^k}$ be such that
$M_0(x)\triangleq x^2+b_0x+c_0$ is irreducible
over $\field_{q^k}$
(such~$b_0$ exists by Corollary~\ref{corollary:irrExists}).
For~$n=2k$, let~$\gamma_0\in \field_{q^n}$ be a root of~$M_0$.
For~$i\in\{0,1,\ldots,\tau{-}1 \}$,
where~$\tau\triangleq\floor{(q{-}1)/2}$,
let $\gamma_i\triangleq w^{i}\gamma_0$ and let
\[
	V_i\triangleq\{u+u^q\gamma_i \,:\, u\in \field_{q^k} \}.
\]
	Finally,
let $\code\triangleq \{\alpha V_i \,:\, i\in\{0,1,\ldots,\tau{-}1\},
\alpha\in\field_{q^n}^* \}$.
\end{construction}

\begin{lemma}\label{lemma:multiplePrbitsProof}
	The
set~$\code$ from~Construction~\ref{construction:multipleOrbits}
is a cyclic subspace code of cardinality~$\tau\cdot (q^n{-}1)/(q{-}1)$
and minimum distance~$2k{-}2$.
\end{lemma}

\begin{proof}
	The fact
that~$\code$
is cyclic follows immediately from its definition. To prove that the minimum distance is~$2k{-}2$, we first show that each~$V_i$ is a Sidon space. To this end,
for~$i\in\{0,1,\ldots,\tau{-}1\}$, let $M_i(x)\triangleq x^2+b_ix+c_i$,
where~$b_i\triangleq w^ib_0$ and~$c_i\triangleq w^{2i}c_0=w^{2i+1}$, and notice
that~$M_i(\gamma_i)=0$.
Moreover, since~$2i{+}1$ is not divisible by~$q{-}1$, it follows that~$c_i\in \overline{W}_{q-1}$. Hence,~$V_i$ is an instance of Construction~\ref{construction:mainNonBinary} and is therefore a Sidon space.	
	 It is left to show that subspaces of \emph{distinct} orbits intersect on a space of dimension at most~$1$; namely, we prove that for any
distinct~$i, j \in \{0,1,\ldots,\tau{-}1\}$
and any~$\alpha\in\field_{q^n}^*$ we have that~$\dim(V_i\cap \alpha V_j)\le 1$. According to Lemma~\ref{lemma:distinctOrbits}, this amounts to showing that any product of a nonzero element of~$V_i$
with a nonzero element of~$V_j$ can be factored uniquely
up to a scalar multiplier from $\field_q$.
For any nonzero~$u$ and~$v$ in~$\field_{q^k}$, we have
	\begin{align*}
		(u+u^q\gamma_i)(v+v^q\gamma_j) &= (u+u^qw^i\gamma_0)(v+v^qw^j\gamma_0)\\
		& = uv+(uv^qw^j+vu^qw^i)\gamma_0+(uv)^qw^{i+j}\gamma_0^2\\
		& = \left(uv-c_0(uv)^qw^{i+j}\right)+\left(uv^qw^j+vu^qw^i-b_0(uv)^qw^{i+j}\right)\gamma_0\\
		& = {\underbrace{\left(uv-(uv)^qw^{i+j+1}\right)}_{\triangleq \, Q_0}}+{\underbrace{\left(uv w^j (v^{q-1}+u^{q-1}w^{i-j})-b_0(uv)^qw^{i+j}\right)}_{\triangleq \, Q_1}}\gamma_0.
	\end{align*}
	Since $0<i{+}j{+}1<q{-}1$, it follows that $w^{i+j+1}\in\overline{W}_{q-1}$. Hence, the mapping~$x\mapsto x-w^{i+j+1}x^q$ is invertible on~$\field_{q^k}$, and~$uv$ may be uniquely
extracted from~$Q_0$. Given~$uv$ and~$Q_1$,
in turn, one can extract the values
of~$P_0\triangleq (u^{q-1}w^{i-j})\cdot v^{q-1}$
and~$P_1\triangleq u^{q-1}w^{i-j}+v^{q-1}$,
and compute the roots of the polynomial~$x^2+P_1 x + P_0$, which
are~$u^{q-1}w^{i-j}$ and~$v^{q-1}$. Since~$0<|i{-}j|<q{-}1$, it follows 
that~$v^{q-1}$ is a~$(q{-}1)$st power in $\field_{q^k}$,
while~$u^{q-1}w^{i-j}$ is not. Hence, by identifying the root which is not a~$(q{-}1)$st power\footnote{Deciding whether a given element~$\mu$ is a~$(q{-}1)$st power can be done by checking if the~$k\times k$ matrix representation (over~$\field_q$) of
the linear mapping~$x\mapsto x^q-\mu x$ is invertible.} and dividing it by~$w^{i-j}$, we find~$u^{q-1}$ and~$v^{q-1}$ and, consequently,~$\{u\field_q \}$ and~$\{v\field_q \}$.
\end{proof}

For odd (respectively, even)~$q$,
the set~$\code$ from
Construction~\ref{construction:multipleOrbits} is a cyclic subspace code
that has cardinality~$(q^n{-}1)/2$
(respectively, $((q{-}2)/(2q{-}2)) \cdot (q^n{-}1)$),
which is within a factor of~$1/2+o_n(1)$
(respectively, $(q{-}2)/(2q{-}2)+o_n(1)$) from the sphere-packing
bound (Theorem~\ref{theorem:EtzionVardyBound}). To the best of our knowledge, this constitutes the first example of a nontrivial cyclic subspace code
of that size.

\subsection{Sidon sets}\label{section:SidonSets}
By taking discrete logarithms from properly chosen elements in a Sidon space~$V\in\grsmn{q}{n}{k}$, a Sidon set in~$\integers_{(q^n-1)/(q-1)}$ is obtained.

\begin{theorem}\label{theorem:SetFromSpace}
	If $V\in\grsmn{q}{n}{k}$ is a Sidon space, $\gamma$ is a primitive element in~$\field_{q^n}$, and
$\{\gamma^{n_i} \,:\, i\in [(q^k{-}1)/(q{-}1)] \}$
is a set of nonzero representatives of all
one-dimensional subspaces of~$V$,
then $\SidonSet \triangleq \{n_i \,:\, i\in [(q^k{-}1)/(q{-}1)] \}$
is a Sidon set in~$\integers_{(q^n-1)/(q-1)}$.
\end{theorem}
\begin{proof}
	Assume that
$a,b,c,d\in \SidonSet$ satisfy
$a+b\equiv c+d \pmod {(q^n{-}1)/(q{-}1)}$,
i.e., $a+b=c+d+t\cdot(q^n{-}1)/(q{-}1)$ for some integer~$t$.
Then,
\[
	\gamma^a\gamma^b = \gamma^c\gamma^d\cdot\lambda,
\]
	where~$\lambda=\gamma^{t(q^n-1)/(q-1)}\in\field_q$. Since~$V$ is a Sidon space, it follows that $\{\lambda\gamma^a\field_{q},\gamma^b\field_{q}\}=\{\gamma^c\field_{q},\gamma^d\field_{q}\}$, i.e., $\{\gamma^a\field_{q},\gamma^b\field_{q}\}=\{\gamma^c\field_{q},\gamma^d\field_{q}\}$. Assume without loss of generality that $\gamma^a\field_{q}=\gamma^c\field_{q}$.
This means that~$\gamma^a$
and~$\gamma^c$ are representatives of the same
one-dimensional subspace of~$V$, so
we must have~$a=c$, thereby concluding the proof.
\end{proof}

The Sidon set which results from applying Theorem~\ref{theorem:SetFromSpace} to Construction~\ref{construction:mainNonBinary} is of particular interest.
Since~$k=n/2$ in this construction, the resulting Sidon set,
$\SidonSet = \SidonSet_q(n) \; (\subseteq \integers_{(q^n-1)/(q-1)})$,
is of size~$(q^{n/2}{-}1)/(q{-}1)$.
For any fixed~$q$, this size is
within a constant factor from the square root of the size of
the domain~$\integers_{(q^n-1)/(q-1)}$; specifically,
\[
\lim_{n \rightarrow \infty}
\frac{|\SidonSet_q(n)|}{|\integers_{(q^n-1)/(q-1)}|^{1/2}}
= \frac{1}{\sqrt{q-1}}
.
\]
As such, the construction $\SidonSet_q(n)$ is optimal,
up to a constant factor (see~\cite[Thm.~5]{SidonSurvey}).

\section{$r$-Sidon spaces and $B_r$-sets}\label{section:rSidon}	

A natural generalization of a Sidon space is an~$r$-Sidon space, in which every product of~$r$ elements can be factored uniquely up to
a scalar multiplier from the base field.

\begin{definition}\label{defintion:rSidonSpace}
	For positive integers~$k<n$ and~$r\ge 2$, a subspace~$V\in\grsmn{q}{n}{k}$ is called an~$r$-Sidon space if for all
nonzero~$a_1,a_2,\ldots,a_r,b_1,b_2,\ldots,b_r\in V$,
if $\prod_{i\in[r]}a_i=\prod_{i\in[r]}b_i$ then the multi-sets~$\{a_i\field_q \}_{i\in[r]}$ and~$\{b_i\field_q \}_{i\in[r]}$ are equal.
\end{definition}

The following lemma provides a sphere-packing
upper bound on the largest possible dimension of an~$r$-Sidon space.

\begin{lemma}\label{lemma:optimalK-rSidon}
	If~$V\in\grsmn{q}{n}{k}$ is an~$r$-Sidon space, then
\[
k < \frac{n}{r} +1+\log_q r .
\]
\end{lemma}

\begin{proof}
	Since the elements in the multi-sets at hand can be seen as
one-dimensional subspaces
of~$V$, it follows that the number of distinct multi-sets is
$\binom{K+r-1}{r}$,
where~$K\triangleq (q^k{-}1)/(q{-}1)$.
By the definition of an~$r$-Sidon space, each such multi-set $\{a_i\field_q\}_{i\in[r]}$ induces a distinct
product~$\prod_{i\in[r]}a_i$ (modulo~$\field_q^*$) and, thus,
\[
	\binom{K+r-1}{r} \le \frac{q^n-1}{q-1}.
\]
	By using the simple bound~$\binom{s}{t}\ge\left(\frac{s}{t}\right)^s$ it follows that
\[
	\left(\frac{K+r-1}{r} \right)^r < q^n
\]
and, so,
\[
	(\log_q K) - (\log_q r)<\log_q(K+r-1)- (\log_q r) < \frac{n}{r}.
\]
	Observing
that $\log_q K = \log_q((q^k{-}1)/(q{-}1)) > k-1$,
the result follows.
\end{proof}

By generalizing
Constructions~\ref{construction:MainBinary}
and~\ref{construction:mainNonBinary},
in what follows two~$r$-Sidon spaces are provided. The first
attains~$k= n/(r{+}1)$
for any~$q$, whereas the second
attains~$k=n/r$ for~$q\ge 3$. 

\begin{construction}\label{construction:rp1Sidon}
	For any field
size~$q$
and integers~$k>0$ and~$r\ge 2$, let~$n=k(r{+}1)$, let~$\gamma\in\field_{q^n}$ be a root of an irreducible polynomial of degree~$r{+}1$ over~$\field_{q^k}$, and
let~$V\triangleq\{u+u^q\gamma \,:\, u\in\field_{q^k} \}$.
\end{construction}

\begin{lemma}\label{lemma:rp1SidonProof}
	The subspace~$V$ from Construction~\ref{construction:rp1Sidon} is an~$r$-Sidon space.
\end{lemma}

\begin{proof}
	We show that given a product~$\prod_{i\in[r]}a_i$ of~$r$ nonzero
elements~$a_1,a_2,\ldots,a_r$
in~$V$, the multi-set~$\{a_i\field_q\}_{i\in[r]}$ can be identified uniquely. To this end,
let~$a_i=u_i+u_i^q\gamma$ for~$i\in[r]$,
and write
\[
	\prod_{i\in[r]}a_i=\prod_{i\in[r]}(u_i+u_i^q\gamma)=P(\gamma),
\]
	where~$P(x)\triangleq\prod_{i\in[r]}(u_i+u_i^qx)$ is a polynomial of degree~$r$ over~$\field_{q^k}$. Since~$\{\gamma^i\}_{i=0}^{r}$ is a linearly independent set over~$\field_{q^k}$, from
the product~$P(\gamma) = \prod_{i\in[r]}(u_i+u_i^q\gamma)$ one can
obtain the coefficients of~$P(x)$. The multi-set of roots of~$P(x)$ is given
by~$\{-1/u_i^{q-1}\}_{i\in[r]}$, from
which one can determine uniquely
the multi-set~$\{u_i\field_q\}_{i\in[r]}$
and, therefore,~$\{a_i\field_q\}_{i\in[r]}$.
\end{proof}

The next construction, which may be seen as a generalization of Construction~\ref{construction:mainNonBinary}, provides an~$r$-Sidon space
with~$k=n/r$
for any~$q\ge 3$. This construction requires an element~$\gamma$ which is a root
of an irreducible polynomial of degree~$r$
over~$\field_{q^k}$ whose free coefficient
is in the set $\overline{W}_{q-1}$, i.e., it is not a~$(q{-}1)$st power of any element in~$\field_{q^k}$. Such an element~$\gamma$ can be obtained as follows. Let~$\beta$ be a primitive element in~$\field_{q^n}$, let~$i$ be an integer in~$[q^n{-}1]$ that is divisible
neither by~$(q^n{-}1)/(q^{kt}{-}1)$
for any proper divisor~$t$ of~$r=n/k$,
nor by~$q{-}1$, and let~$\gamma=-\beta^i$.

\begin{lemma}\label{lemma:Pgamma}
	The minimal polynomial~$M(x)$ of~$\gamma$ over~$\field_{q^k}$
is of degree~$r$ and satisfies $M(0)\in\overline{W}_{q-1}$.
\end{lemma}

\begin{proof}
	By the choice of~$i$, the element~$\beta^i$
belongs neither to $\field_{q^k}$ not to any field extension of it
that is a proper subfield of~$\field_{q^n}$. Hence,~$\deg M(x)=r$ and
\[
	M(0)
	= \prod_{j=0}^{r-1}(-\gamma)^{q^{jk}}
	= (-\gamma)^{(q^n-1)/(q^k-1)}=\beta^{i(q^n-1)/(q^k-1)}
	= w^i,
\]
	where~$w\triangleq\beta^{(q^n-1)/(q^k-1)}$ is a primitive element in~$\field_{q^k}$. Since~$i$ is not divisible by~$q{-}1$, we get
that~$M(0)=w^i\in\overline{W}_{q-1}$.
\end{proof}

\begin{construction}\label{construction:rSidonNonBinary}
	For any prime power~$q\ge 3$ and integers~$k>0$ and~$r\ge 2$, let~$n=kr$, let~$\gamma\in\field_{q^n}$ be a root of an irreducible polynomial~$M(x)$ of degree~$r$ over~$\field_{q^k}$ with~$M(0)\in\overline{W}_{q-1}$, and
let~$V\triangleq \{u+u^q\gamma \,:\, u\in\field_{q^k} \}$.
\end{construction}

\begin{lemma}\label{lemma:rSidonNonBinary}
	The subspace~$V$ from Construction~\ref{construction:rSidonNonBinary} is an~$r$-Sidon space.
\end{lemma}

\begin{proof}
	Following the same methodology and the same notation from the proof of Lemma~\ref{lemma:rp1SidonProof}, notice that
for~$a_1,a_2,\ldots,a_r\in V$,
where~$a_i=u_i+u_i^q\gamma$ for~$i\in[r]$
and~$u_1,u_2,\ldots,u_k\in\field_{q^k}^*$,
\[
	\prod_{i\in[r]}a_i=\prod_{i\in[r]}(u_i+u_i^q\gamma)
	= P(\gamma)
	= {\underbrace{(P(x)- p_r\cdot M(x))}_{\triangleq \, Q(x)}}
		\raisebox{-0.7ex}{$\bigm|_{x = \gamma}$} ,
\]
	where~$p_r=\prod_{i\in[r]}u_i^q=(P(0))^q$
is the leading coefficient of~$P(x)$.
	
	Since~$\{\gamma^i\}_{i=0}^{r-1}$ is a linearly independent set over~$\field_{q^k}$, from the product~$\prod_{i\in[r]}(u_i+u_i^q\gamma)$ one can determine the coefficients
of~$Q(x)=P(x)-p_r M(x)$,
which is a polynomial of degree less than~$r$ over~$\field_{q^k}$.
Similarly to the proof of Theorem~\ref{theorem:mainNonBinary}, the free
coefficient of~$Q(x)$, which is given by
\[
Q(0) = P(0) - p_r M(0)= P(0) - M(0)(P(0))^q ,
\]
is the evaluation of the linearized
polynomial~$T(x)\triangleq x - M(0)x^q$
at the point~$x=P(0)$.
	Since~$M(0)\in\overline{W}_{q-1}$,
it follows that~$x\mapsto T(x)$ is invertible, and $P(0)$ can be
determined uniquely. Knowing~$P(0)$,
$p_r=(P(0))^q$, and~$Q(x)$, we can
compute~$P(x)=Q(x) + p_r M(x)$, whose multi-set of
roots, $\{-1/u_j^{q-1}\}_{i\in[r]}$,
determines uniquely the multi-set~$\{a_i\field_q\}_{i\in[r]}$.
\end{proof}

Similarly to Theorem~\ref{theorem:SetFromSpace},
the~$r$-Sidon space constructions in this section provide
constructions of the following generalization of Sidon sets.

\begin{definition}[\cite{SidonSurvey}]
	A subset $\SidonSet$ of an Abelian group~$G$ is
called a~$B_r$-set if the sums of all multi-sets of size~$r$ of
elements from~$\SidonSet$ are distinct. 
\end{definition}

It is known that the size~$R_r(m)$ of the largest~$B_r$-set in~$[m]$ satisfies (see~\cite[Sec.~4.2]{SidonSurvey})
\begin{equation}\label{equation:BrSize}
	1\le \lim_{m\to\infty}\frac{R_r(m)}{\sqrt[r]{m}}\le \frac{1}{2e}\left(r+\frac{3}{2} \log r+o_r(\log r) \right).
\end{equation}

\begin{theorem}\label{theorem:rSidonSet}
	For an $r$-Sidon space~$V\in\grsmn{q}{n}{k}$, and for
a set~$\{\gamma^{n_i} \,:\, i\in [(q^k{-}1)/(q{-}1)] \}$
of nonzero representatives of all one-dimensional subspaces
of~$V$ for some primitive~$\gamma\in\field_{q^n}$,
the set~$\SidonSet \triangleq\{n_i \,:\, i\in [(q^k{-}1)/(q{-}1)] \}$
is a~$B_r$-set in~$\integers_{(q^n-1)/(q-1)}$.
\end{theorem}

The proof of Theorem~\ref{theorem:rSidonSet}
is a straightforward generalization of the proof
of Theorem~\ref{theorem:SetFromSpace} (and is therefore omitted).
Applying
Theorem~\ref{theorem:rSidonSet} to
Construction~\ref{construction:rSidonNonBinary} results in
a~$B_r$-set $\SidonSet_q(n,r)$ of
size $(q^{n/r}{-}1)/(q{-}1)$
in~$\integers_{(q^n-1)/(q-1)}$
(and therefore in~$[(q^n{-}1)/(q{-}1)]$). It is readily verified that 
\[
\lim_{n\to \infty}\frac{|\SidonSet_q(n,r)|}{%
          |\integers_{(q^n-1)/(q-1)}|^{1/r}}
=(q-1)^{(1/r)-1}
\]
(compare with~\eqref{equation:BrSize}).

\section*{Acknowledgments}
The authors would like to express their sincere gratitude to Amir Shpilka for valuable discussions.

\appendices

\section{Sidon spaces of dimensions~$1$ and~$2$}
\label{section:SidonOfDimOneTwo}

Definition~\ref{definition:SidonSpace} readily implies that any subspace
in $\grsmn{q}{n}{1}$ is a Sidon space. The same holds
for~$\grsmn{q}{n}{2}$, as stated next.

\begin{lemma}\label{lemma:SidonDim2}
	Any subspace~$V\in\grsmn{q}{n}{2}$ is either a max-span Sidon space or---if $n$ is even---a cyclic shift of~$\field_{q^2}$.
\end{lemma}

\begin{proof}
	Since the Sidon space property is invariant under cyclic shifts, it can be assumed without loss of generality that~$V=\Span{1,v}$ for some~$v\in\field_{q^n}\setminus \field_q$. This implies
that~$V^2=\Span{1,v,v^2}$ and, hence,~$\dim(V^2)\in\{2,3\}$.
If~$\dim(V^2)=3$ then~$V$ is a max-span Sidon space (see Definition~\ref{definition:minmax-spanSidonSpace}). On the other hand, if~$\dim(V^2)=2$, then~$v^2=\lambda v+\mu$
for some~$\lambda$ and~$\mu$
in~$\field_q$, which means that $v \in \field_{q^2}$
(and, so,~$n$ must be even); hence, $V=\field_{q^2}$ in this case.
\end{proof}

\section{Proof of Theorem~\protect\ref{theorem:Existence}}
\label{section:SidonForAnyN}

Constructions~\ref{construction:MainBinary}
and~\ref{construction:mainNonBinary}
provide Sidon spaces in which~$n$ is linear in~$k$. However, these
constructions are restricted to integers~$n$ that have a large proper
divisor, and do not apply in many other cases (e.g., when~$n$ is
a prime).
We show here that Sidon spaces
exist for \emph{any}~$n$ and $k \le (n{-}2)/4$.
This claim is proved in an inductive
manner: given a Sidon space~$V$,
we first show that as long as $\field_{q^n}\setminus V$ is large enough,
a new element~$v$ can
be added to~$V$ without compromising the Sidon space property.

\begin{lemma}\label{lemma:ExistenceInduction}
	If~$V\in\grsmn{q}{n}{k}$ is a Sidon space and
$q^n-q^k>2\cdot (q^{4k+4}{-}1)/(q{-}1)$,
then there exists $v\in\field_{q^n}\setminus V$ such that $V+\Span{v}$ is a Sidon space as well.
\end{lemma}

\begin{proof}
	For $v\in\field_{q^n}\setminus V$, the elements of $V+\Span{v}$ are of the form $\alpha v+a$, where $\alpha\in\field_q$ and $a\in V$. Hence, in order for $V+\Span{v}$ to be a Sidon space, it suffices to prove that for every $\alpha_1,\alpha_2, \alpha_3,\alpha_4\in\field_q$ and every $a,b,c,d\in V$ such that none of $\alpha_1v+a,~\alpha_2v+b,~\alpha_3v+c$, and $\alpha_4v+d$ is zero, we have that
	\begin{align}
	\label{eqn:U+vCondition1}\textrm{if } \left(\alpha_1v+a\right)\left(\alpha_2v+b\right)&=\left(\alpha_3v+c\right)\left(\alpha_4v+d\right),\textrm{ then}\\
	\label{eqn:U+vCondition2}\alpha_1v+a&\in
\Span{\alpha_3v+c} \cup \Span{\alpha_4v+d}.
	\end{align}
	Equation~\eqref{eqn:U+vCondition1}
is equivalent to~$v$ being a solution to the quadratic equation
	\begin{equation}\label{equation:xEquality}
	\left(\alpha_1\alpha_2-\alpha_3\alpha_4\right)x^2+\left(\alpha_1b+\alpha_2a-\alpha_3d-\alpha_4c\right)x+ab-cd= 0,
	\end{equation}
	which is trivial (i.e.,
the left-hand side is the zero polynomial) only if
$\left(\alpha_1x+a\right)\left(\alpha_2x+b\right)$ and
$\left(\alpha_3x+c\right)\left(\alpha_4x+d\right)$
are (irreducible) decompositions of the same quadratic polynomial
(i.e., only if~\eqref{eqn:U+vCondition2} holds).
Therefore, the total number of solutions for~$x$ to all nontrivial equations of the form~\eqref{equation:xEquality} serves as an upper bound
on
the number of elements in~$\field_{q^n}\setminus V$ that
\emph{cannot}
be added to the Sidon space~$V$ while maintaining the Sidon
space property.  Since each such equation
is determined
by four elements in~$\field_q$ and four elements in~$V$, it follows that there exists
at most~$(q^{4k+4}{-}1)/(q{-}1)$ such equations that are
pairwise linearly independent over $\field_q$.
As we assume that
$q^n-|V| = q^n - q^k >2\cdot (q^{4k+4}{-}1)/(q{-}1)$,
there exists $v\in\field_{q^n}\setminus V$ that
satisfies~\emph{all} constraints of
the form~\eqref{eqn:U+vCondition1}--\eqref{eqn:U+vCondition2}
and, hence,~$V+\Span{v}$
is a Sidon space.
\end{proof}

\begin{proof}[Proof of Theorem~\ref{theorem:Existence}]
	We prove by induction on $k = 1, 2, \ldots, \floor{(n{-}2)/4}$
that $\grsmn{q}{n}{k}$ contains a Sidon space, with the induction base
($k = 1$) being straightforward.

Turning to the induction step, suppose that $\grsmn{q}{n}{k}$
contains a Sidon space, for some $k \le \floor{(n{-}6)/4}$.
By Lemma~\ref{lemma:ExistenceInduction}, any Sidon space
in~$\grsmn{q}{n}{k}$ can be expanded to a Sidon space
in~$\grsmn{q}{n}{k{+}1}$, as long as
\[
	q^n > \frac{2}{q-1} \cdot (q^{4k+4}-1) + q^k 
\]
which, in turn, holds if
\[
	q^{n-k} \ge \frac{2}{q-1} \cdot q^{3k+4}+1.
\]
It is easy to see that the latter inequality is implied by
$n-k \ge 3k+6$, or, equivalently,
by our induction assumption~$k \le (n{-}6)/4$. Hence,
for such~$k$, any Sidon space in~$\grsmn{q}{n}{k}$
can be expanded to a Sidon space in~$\grsmn{q}{n}{k{+}1}$.
\end{proof}

\section{Analysis of Construction~\ref{construction:irreduciblePolys}}
\label{section:omittedProofs}

\begin{proof}[Proof of Lemma~\ref{lemma:irreduciblePolys}]
First, note that since the largest degree of a polynomial
in~$\Irred$ is~$\Delta$, it follows 
that~$\deg f_i\le \Delta \cdot \binom{k}{2} < n/2$, for all~$i\in[k]$.
Now, by Lemma~\ref{lemma:max-spanSidonSpace},
it suffices to prove that
the set~$F_\gamma\triangleq\{f_i(\gamma)\cdot f_j(\gamma) \,:\,
i,j\in[k], i\ge j\}$
is linearly independent over~$\field_q$. From
$n/2 > \max\{\deg f_i \,:\, i \in [k] \}$
it follows that the set of \emph{field elements} $F_\gamma$
is linearly independent over $\field_q$,
if and only if the set of \emph{polynomials}
$F_x\triangleq\{{f_i(x)\cdot f_j(x)} \,:\, {i,j\in[k]}, i\ge j\}$
is linearly independent over $\field_q$.

	Assume that
	\begin{equation}\label{equation:linearDependence}
	\sum_{i,j\in[k]\,:\,i \ge j}\alpha_{i,j}f_i(x)f_j(x)=0,
	\end{equation}
for coefficients~$\alpha_{i,j} \in \field_q$.
According to~\eqref{equation:Aifi's},
taking~\eqref{equation:linearDependence} modulo~$p_{s,s}(x)$
for any $s \in [k]$ results in
\[
	\alpha_{s,s}(f_s(x))^2 \equiv 0 \pmod {p_{s,s}(x)} .
\]
	Since~$\gcd(f_s(x),p_{s,s}(x)) = 1$,
it follows that~$\alpha_{s,s}=0$ for all~${s\in[k]}$,
and~\eqref{equation:linearDependence} becomes
	\begin{equation}\label{equation:linearDependence2}
	\sum_{i,j\in[k]\,:\,i > j}\alpha_{i,j}f_i(x)f_j(x)=0 .
	\end{equation}
Taking now~\eqref{equation:linearDependence2} modulo~$p_{s,t}(x)$
for any $s > t$ in $[k]$ yields
	\[
	\alpha_{s,t}f_s(x)f_t(x) \equiv 0 \pmod {p_{s,t}(x)} .
	\]
	Again, $\gcd(f_s(x),p_{s,t}(x)) = \gcd(f_t(x),p_{s,t}(x)) = 1$
and, so, $\alpha_{s,t}=0$ for all~$s > t$ in $[k]$.
\end{proof}

To evaluate the contribution of Construction~\ref{construction:irreduciblePolys},
we provide an upper bound on the smallest possible largest
degree, $\Delta$, of the elements of~$\Irred$, under the constraint
that $|\Irred| = \binom{k+1}{2}$.

\begin{lemma}\label{lemma:irrPolyEstimation}
For any positive integer~$\ell$, the number~$J(\ell)$
of monic irreducible polynomials of degree at most~$\ell$
over $\field_q$ satisfies $J(\ell) \ge q^\ell/\ell$.
\end{lemma}
\begin{proof}
	Let $N(\ell)$ denote the number of monic irreducible polynomials
of degree (exactly)~$\ell$ over~$\field_q$. 
It is known that
\[
\sum_{t | \ell} t \, N(t) = q^\ell
\]
(see~\cite[Ch.~3, Cor.~3.21]{FiniteFields}). Therefore,
\[
J(\ell) = \sum_{t \in [\ell]} N(t) \ge
\frac{1}{\ell} \sum_{t | \ell} t \, N(t) = \frac{q^\ell}{\ell} .
\]
\end{proof}

Given~$q$ and~$k$, we can select~$\Delta$
in Construction~\ref{construction:irreduciblePolys}
to be the smallest for which
$J(\Delta) \ge |\Irred| = \binom{k+1}{2}$.
By Lemma~\ref{lemma:irrPolyEstimation} we see
that $\Delta = (2 + o_k(1)) \log_q k$ will do,
in which case we can take $n = (2 + o_k(1)) k^2 \log_q k$.
In particular,
when~$q \ge |\Irred| = \binom{k+1}{2}$ we can take~$\Delta = 1$,
which yields a construction for any $n \ge k(k{-}1){+}1$.

\section{Proof of Theorem~\protect\ref{theorem:ExistenceMaxSpan}}
\label{section:ExistenceMaxSpan}

In our proof of Theorem~\ref{theorem:ExistenceMaxSpan}, we will borrow
tools from~\cite{Cascudo};
specifically, we will use properties
of quadratic forms over finite fields,
as found in~\cite[Ch.~6, Sec.~2]{FiniteFields}
and~\cite[Sec.~IV and Appendix~A]{Cascudo}.

A \emph{quadratic form} (in~$k$ indeterminates)
over a field~$F$ is a homogeneous polynomial of the form
\[
Q(\bldx) = \sum_{i,j \in [k] \,:\, i \ge j} a_{i,j} x_i x_j ,
\]
where $a_{i,j} \in F$ and
$\bldx = (x_1 \; x_2 \; \ldots \; x_k)$ is a vector of
indeterminates.
The \emph{rank} of $Q(\bldx)$ equals
the smallest number
of indeterminates that will actually appear in $Q(\bldx P)$,
when~$P$ ranges over all nonsingular $k \times k$ matrices
over~$F$. 

Through the canonical representation of~$Q(\bldx)$
(as in~\cite[Thms~6.21 and~6.30]{FiniteFields}),
it readily follows that the rank of $Q(\bldx)$ is the same in any
extension field of~$F$.
The set of all quadratic forms in~$k$ indeterminates
with rank~$r$ over $\field_q$ will be denoted by $\Quad_q(k,r)$,
and we will use the shorthand notation~$\Quad_q(k)$ for~$\Quad_q(k,k)$.

The next two lemmas are taken from~\cite[Thms.~4.5 and~4.6]{Cascudo}
(see also~\cite[Ch.~6]{FiniteFields}).

\begin{lemma}
\label{lemma:RootsQ}
Given a prime power~$q$
and positive integers~$n$, $k$, and $r \in [k]$,
for any~$Q \in \Quad_q(k,r)$,
the number of vectors $\bldv \in \field_{q^n}^k$
that satisfy $Q(\bldv) = 0$ is given by
\[
\left\{
\begin{array}{lcl}
q^{n(k-1)} && \textrm{if $r$ is odd} \\
q^{n(k-1)} \cdot \left( 1 \pm  (q^n{-}1) \cdot q^{-rn/2} \right)
&& \textrm{if $r$ is even} \\
\end{array}
\right.
.
\]
\end{lemma}

\begin{lemma}
\label{lemma:NumberQ}
Given a prime power~$q$ and a positive integer~$k$,
\[
|\Quad_q(k)| =
q^{k(k+1)/2} \cdot \prod_{j \in [\ceil{k/2}]}
\left( 1 - q^{1-2j} \right)
\]
and, for any~$r \in [k{-}1]$,
\[
\left| \Quad_q(k,r) \right|
= \qbin{k}{r}{q} \cdot \left| \Quad_q(r) \right| .
\]
\end{lemma}

We will also use the following bound on
the $q$-binomial coefficients.

\begin{lemma}[{\cite[Lemma~4]{CodingFor}}]
\label{lemma:qbin}
For any prime power~$q$ and integers~$t \ge s \ge 0$,
\[
\qbin{t}{s}{q} < 4 \cdot q^{s(t-s)} .
\]
\end{lemma}

\begin{proof}[Proof of Theorem~\protect\ref{theorem:ExistenceMaxSpan}]
For $k = 1$, every subspace in $\grsmn{q}{n}{k}$ is
a max-span Sidon space; hence we assume from now on in the proof that
$k \ge 2$. Let $\xi_1, \xi_2, \ldots, \xi_k$
be elements that are uniformly and independently
selected from~$\field_{q^n}$.
We bound from above the probability that 
the set $\{ \xi_i \}_{i \in [k]}$ does not span
a $k$-dimensional max-span Sidon space.
That probability bounds from above the fraction of
the spaces in $\grsmn{q}{n}{k}$ that are not max-span Sidon spaces.

Write $\bldxi = (\xi_1 \; \xi_2 \; \ldots \; \xi_k)$
and, for $r \in [k]$, let
$\event_r$ denote the event that
$Q(\bldxi) = 0$ for some $Q \in \Quad_q(k,r)$.
Then $\cup_{r \in [k]} \event_r$ stands for
the event that $\bldxi$ is not a max-span Sidon space.
By a union bound, we have
\begin{equation}
\label{equation:union}
\Prob \left\{ \textstyle\cup_{r \in [k]} \event_r \right\} \le
\Prob \left\{ \event_1 \cup \event_2 \right\} +
\sum_{r=3}^k \Prob \left\{ \event_r \right\} .
\end{equation}
Next, we bound from above the terms
in the right-hand side of~\eqref{equation:union}.

Starting with~$\event_1$, this event is equivalent
to having $\bldxi \cdot \blda^\transpose = 0$, for some nonzero
$\blda \in \field_q^k$; namely, it is equivalent to
$\{ \xi_i \}_{i \in [k]}$ being a linearly dependent set
over~$\field_q$. Turning to~$\event_2$, shifting to canonical
quadratic forms, as in~\cite[Thms.~6.21 and~6.30]{FiniteFields},
it follows that this event is equivalent to~$\bldxi$
satisfying

\begin{equation}
\label{equation:rootquad}
p_0 \cdot (\bldxi \cdot \blda^\transpose)^2
+ p_1 \cdot
(\bldxi \cdot \blda^\transpose) (\bldxi \cdot \bldb^\transpose)
+ p_2 \cdot (\bldxi \cdot \bldb^\transpose)^2 = 0 ,
\end{equation}
for some linearly independent vectors $\blda, \bldb \in \field_q^k$ over~$\field_q$ and a nonzero~$(p_0 \; p_1 \; p_2) \in \field_q^3$.
Since~\eqref{equation:rootquad} is equivalent to having
$(\bldxi \cdot \blda^\transpose)/(\bldxi \cdot \bldb^\transpose)
\in \field_{q^2}$, it follows that~$\event_2$ implies that
$\{ \xi_i \}_{i \in [k]}$ is a linearly dependent set
over\footnote{%
When~$n$ is odd, linear dependence over~$\field_{q^2}$ is the same
as linear dependence over~$\field_q$.}
$\field_{q^2}$. We conclude that
$\event_1 \cup \event_2$ implies that
$\bldxi \cdot \blda^\transpose = 0$,
for some nonzero $\blda \in \field_{q^2}^k$
whose leading nonzero coefficient is~$1$ (say).
Thus,
\begin{equation}
\label{equation:event1+2}
\Prob \left\{ \event_1 \cup \event_2 \right\}
\le \frac{q^{2k} - 1}{(q^2{-}1) \cdot q^n}
= \frac{q^k-1}{(q{-}1) \cdot q^n} \cdot \frac{q^k+1}{q+1} ,
\end{equation}
which proves the theorem for~$k = 2$
(see also Lemma~\ref{lemma:SidonDim2}
in Appendix~\ref{section:SidonOfDimOneTwo}).
Hence, we assume hereafter that~$k \ge 3$.

In the sequel, we will need a lower bound
on the following difference:
\begin{eqnarray}
\lefteqn{
\frac{|\Quad_q(k,1)| + |\Quad_q(k,2)|}{q^n(q-1)}
- \Prob \left\{ \event_1 \cup \event_2 \right\}
} \makebox[2ex]{}
\nonumber \\
& \stackrel{
\eqref{equation:event1+2}, \,\mathrm{Lemma~\ref{lemma:NumberQ}}}{\ge} &
\frac{1}{q^n(q{-}1)}
\left(
\qbin{k}{1}{q} \cdot \left| \Quad_q(1) \right| +
\qbin{k}{2}{q} \cdot \left| \Quad_q(2) \right|
- (q^k-1) \cdot \frac{q^k+1}{q+1}
\right) 
\nonumber \\
& \stackrel{\mathrm{Lemma~\ref{lemma:NumberQ}}}{=} &
\frac{1}{q^n(q{-}1)}
\left(
q^k - 1 +
\frac{q^2 (q^k - 1)(q^{k-1} - 1)}{q^2 - 1}
- (q^k-1) \cdot \frac{q^k+1}{q+1}
\right)
\nonumber \\
& = &
\frac{q^k - 1}{q^n(q^2{-}1)} \cdot \frac{q^k - q}{q-1}
\nonumber \\
\label{equation:diff1+2}
& \stackrel{k \ge 3}{>} &
\frac{q^k - 1}{q^n(q{-}1)} .
\end{eqnarray}

Continuing now with the right-hand side of~\eqref{equation:union},
we have:
\begin{eqnarray*}
\Prob \left\{ \textstyle\cup_{r \in [k]} \event_r \right\}
& \stackrel{\mathrm{Lemma~\ref{lemma:RootsQ}}}{\le} &
\Prob \left\{ \event_1 \cup \event_2 \right\} \\
&& \quad {} +
\frac{q^{n(k-1)}}{q^{nk}(q{-}1)}
\Bigl(
{\sum_{r=3}^k 
|\Quad_q(k,r)|}
+
{\sum_{i=2}^{\floor{k/2}}
|\Quad_q(k,2i)| \cdot (q^n{-}1) \cdot q^{-in}}
\Bigr) \\
& \stackrel{\eqref{equation:diff1+2}}{<} &
\frac{1}{q^n(q{-}1)}
\Bigl(
{\underbrace{\sum_{r \in [k] }
|\Quad_q(k,r)|}_{q^{k(k+1)/2}-1}}
- (q^k - 1)
+
(q^n{-}1)
\sum_{i=2}^{\floor{k/2}}
|\Quad_q(k,2i)| \cdot q^{-in}
\Bigr) \\
& = &
\frac{q^{k(k+1)/2}-q^k}{q^n(q{-}1)}
+
\frac{q^n-1}{q^n(q{-}1)}
{\underbrace{
\sum_{i=2}^{\floor{k/2}} |\Quad_q(k,2i)| \cdot q^{-in}
}_{\triangleq \, \varepsilon(n)}} \\
& < &
\frac{q^{k(k+1)/2-n}}{q{-}1}
+
\frac{1}{q{-}1}
\left( \varepsilon(n) - q^{k-n} \right) .
\end{eqnarray*}
Hence, in order to complete the proof, it suffices
to show that $\varepsilon(n) \le q^{k-n}$.
Indeed, for $k = 3$ we have $\varepsilon(n) = 0$;
otherwise, for $k \ge 4$,
\begin{eqnarray*}
\varepsilon(n)
& = &
\sum_{i=2}^{\floor{k/2}} |\Quad_q(k,2i)| \cdot q^{-in} \\
& \stackrel{\mathrm{Lemma~\ref{lemma:NumberQ}}}{=} &
\sum_{i=2}^{\floor{k/2}}
\qbin{k}{2i}{q} \cdot
{\underbrace{|\Quad_q(2i)|}_{< \, q^{i(2i+1)}}}
\cdot q^{-in} \\
\\
& \stackrel{\mathrm{Lemma~\ref{lemma:qbin}}}{<} &
4 \cdot
\sum_{i=2}^{\floor{k/2}} q^{2i(k-2i)+i(2i+1)-in} \\
& = &
4 \cdot q^{4k-2n-6}
\sum_{i=0}^{\floor{k/2}-2} q^{i(2k-n)-(2i+7)i} \\
& < &
4 \cdot q^{4k-2n-6}
{\underbrace{\sum_{i=0}^{\infty} q^{i(2k-n)}}_{< \, 2}} \\
& < &
q^{4k-2n-3} < q^{k-n} ,
\end{eqnarray*}
whenever~$k \ge 4$ and~$n \ge \binom{k+1}{2}$.
\end{proof}

\begin{remark}
For $q = 2$ and $n = \binom{k+1}{2}$,
our analysis does not rule out
the possibility that the fraction of max-span Sidon spaces
within~$\grsmn{q}{n}{k}$ is 
exponentially small in~$n$. However, empirical results
suggest that, as $k$ increases, this fraction converges---%
for the tested values of~$q$---%
to (approximately)~$\prod_{i=1}^\infty (1 - q^{-i})$,
similarly to the fraction of invertible matrices among all
$k \times k$ matrices over $\field_q$;
for~$q = 2$, this limit is approximately~$0.288$.
Our analysis herein was too crude to capture this behavior,
and a proof that the said fraction indeed
converges to this limit is yet to be found.\qed
\end{remark}

\begin{thebibliography}{10}
	\bibitem{Vospers}
		C.~Bachoc, O.~Serra, and G.~Z\'{e}mor, ``An analogue of Vosper's theorem for extension fields,'' \emph{arXiv:1501.00602 [math.NT]}, 2015.
		
	\bibitem{SubspacePolynomials}
		E.~Ben-Sasson, T.~Etzion, A.~Gabizon, and N.~Raviv, ``Subspace polynomials and cyclic subspace codes,'' \emph{IEEE Transactions on Information Theory}, vol.~62, no.~3, pp.~1157--1165, 2016.
		
	\bibitem{Bose}
		R.~C.~Bose and S.~Chowla, ``Theorems in the additive theory of numbers,'' \emph{Commentarii Mathematici Helvetici}, vol.~37, no.~1, pp.~141--147, 1962.

        \bibitem{Cascudo}
                I.~Cascudo, R.~Cramer, D.~Mirandola, and G.~Z\'{e}mor,
                ``Squares of random linear codes,''
                \emph{IEEE Transactions on Information Theory},
                vol.~61, no.~3, pp.~1159--1173, 2015.

	\bibitem{EtzionVardy}
		T.~Etzion and A.~Vardy, ``Error-correcting codes in projective space,'' \emph{IEEE Transactions on Information Theory}, vol.~57, no.~2, pp.~1165--1173, 2011.
		
	\bibitem{GLMT}
		H.~Gluesing-Luerssen, K.~Morrison, and C.~Troha,
		``Cyclic orbit codes and stabilizer subfields,''
    		\emph{Advances in Mathematics of Communications},
		vol.~9, pp.~177--197, 2015.

	\bibitem{Golomb}
		S.~W.~Golomb. \emph{Shift Register Sequences}. Aegean Park Press, 1982.
		
	\bibitem{CodingFor}
		R.~Koetter and F.~R.~Kschischang, ``Coding for errors and erasures in random network coding,'' \emph{IEEE Transactions on Information Theory}, vol.~54, no.~8, pp.~3579--3591, 2008.
		
	\bibitem{FiniteFields}
		R.~Lidl and H.~Niederreiter, \emph{Finite fields}. Cambridge University Press, 1997.
	\bibitem{Niederreiter}
    		H.~Niederreiter,
    		``An enumeration formula for certain irreducible
		polynomials with an application to the construction of
		irreducible polynomials over the binary fields,''
    		\emph{Applicable Algebra in Engineering,
		Communication and Computing},
		vol.~1, pp.~119--124, 1990.

	\bibitem{SidonSurvey}
		K.~O'Bryant, ``A complete annotated bibliography of work related to Sidon sequences,'' \emph{Electronic Journal of Combinatorics}, vol.~575, 2004.
	
	\bibitem{Otal}
		K.~Otal and F.~{\"O}zbudak, ``Cyclic subspace codes via subspace polynomials,'' \emph{Designs, Codes and Cryptography}, doi:10.1007/s10623-016-0297-1, 2016.
		
	\bibitem{Ronny}
		R.~M.~Roth, \emph{Introduction to Coding Theory}. Cambridge University Press, 2006.
		
	\bibitem{CyclicOrbitCodes}
		A.~L.~Trautmann, F.~Manganiello, M.~Braun, and J.~Rosenthal, ``Cyclic orbit codes,'' \emph{IEEE Transactions on Information Theory} vol.~59, no.~11, pp.~7386--7404, 2013.
		
	\end{thebibliography}
\end{document}